\documentclass[11pt,a4paper]{article}

\usepackage[utf8]{inputenc}

\usepackage[english]{babel}

\normalfont

\usepackage{amsmath}

\usepackage{amssymb}

\usepackage{amsthm}

\usepackage{amsfonts}

\usepackage{mathabx}

\usepackage{bbm}

\usepackage{bm}

\usepackage{mathrsfs}

\usepackage{pifont}

\usepackage{hyperref}

\usepackage{pgf}

\usepackage{graphicx}

\usepackage[inline]{enumitem}

\usepackage{tikz-cd}

\setlist{nosep, itemsep=.1cm, topsep=.1cm}

\newcommand{\be}[0]{\begin{equation}}

\newcommand{\ee}[0]{\end{equation}}

\numberwithin{equation}{section}

\theoremstyle{plain}% default

\newtheorem{theorem}{Theorem}[section]

\newtheorem{lemma}[theorem]{Lemma}

\theoremstyle{definition}

\begin{document}

\vspace*{-1cm}

\thispagestyle{empty}

\vspace*{1.5cm}

\begin{center}

{\LARGE 

{\bf Higher abelian Dijkgraaf-Witten theory}}

\vspace{2.0cm}

{\large Samuel Monnier}

\vspace*{0.5cm}

Institut für Mathematik, %Mathematisch-naturwissenschaftliche Fakultät, 
Universit\"at Z\"urich,\\
Winterthurerstrasse 190, 8057 Z\"urich, Switzerland\\
samuel.monnier@gmail.com\\

\vspace*{1cm}

{\bf Abstract}

\end{center}

Dijkgraaf-Witten theories are quantum field theories based on (form degree 1) gauge fields valued in finite groups. We describe their generalization based on $p$-form gauge fields valued in finite abelian groups, as field theories extended to codimension 2.\\

%\vspace{.3cm}

\noindent MSC class: 81T45

%\newpage

%\tableofcontents

\section{Introduction and summary}

%Classical abelian Chern-Simons theories can be generalized by replacing the abelian gauge field, represented locally by a 1-form or globally by a degree 2 differential cocycle, by a degree $p+1$ differential cocycle. These higher classical Chern-Simons theories are for instance crucial ingredients in the construction of many supergravity theories.

Dijkgraaf-Witten (DW) theories \cite{Dijkgraaf:1989pz} are essentially Chern-Simons theories for gauge fields valued in a finite group $\Gamma$, and can be defined in any dimension. Their fields are connections on principal $\Gamma$-bundles. Due to the finiteness of $\Gamma$, there is only one connection on each principal bundle and it is necessarily flat. As a result, the space of fields is finite, and the path integral reduces to a finite instanton sum, making their exact quantization straightforward. For this reason, they are interesting toy models of quantum gauge field theories. 

Abelian gauge fields have higher degree cousins, described locally by $p$-forms and globally by degree $p+1$ differential cohomology classes \cite{hopkins-2005-70}. When the gauge group is $U(1)$, they can be thought of as connections on certain "higher circle bundles" that can be defined using higher category theory. We describe in the present paper generalizations of abelian Dijkgraaf-Witten theories whose fields are higher gauge fields valued in a finite abelian group $\Gamma$. Just as for ordinary DW theories, the path integrals are finite sums and we can describe the quantum theories exactly. 

%Dijkgraaf-Witten theory is usually formulated in terms of We do not have a good picture for what an abelian higher gauge field valued in $\Gamma$ should be, but fortunately we know that the gauge equivalence classes of abelian higher gauge fields of degree $p$ on a manifold $M$ should be classified by $H^p(M; \Gamma)$, as explained in Section \ref{SecGamGerb}. As the Dijkgraaf-Witten theory is really a gauge theory, it depends only on the gauge equivalence classes of the fields, and not on the pa

Two subtleties appear in the construction below. The first is about finding a good model for the higher gauge fields. We do not know a convenient higher generalization of principal bundles with connection valued in a finite group (but see \cite{Baez:2004in, Schreiber2011, Zucchini:2011aa, Jurco:2014mva} for the $p = 2$ case). However, as Dijkgraaf-Witten theory is a gauge theory, only the set of isomorphism classes of fields matters. On a manifold $M$, the isomorphism classes of higher abelian gauge fields are given by $H^p(M; \Gamma)$, and we can take the fields to be cocycles valued in $\Gamma$. Indeed, ordinary Dijkgraaf-Witten theories themselves can be reformulated in terms of 1-cocycles valued in $\Gamma$, instead of principal $\Gamma$-bundles.

The second subtlety is the determination of the measure \eqref{EqDefNormFact1} on the space of fields, which appears in the instanton sum defining the quantum theory. These factors crucially obey the relation \eqref{EqRelMeasAutGroups}, which ensures that the field theory functor is compatible with the gluing of manifolds with boundary, as we show in Section \ref{SecGlu}. Our restriction to abelian groups makes the measure constant across the space of fields, leaving only the dependence on the underlying manifold. The structure of the measure is nevertheless interesting, being given by an alternating product of orders of $\Gamma$-valued cohomology groups. It suggests an interpretation in terms of a tower of ghosts that is not made explicit in our construction. In more mathematical terms, it should coincide with the homotopy cardinality of a $p$-groupoid of fields, but we will not attempt to make this higher categorical structure manifest here.% In the case of higher abelian gauge fields valued in $U(1)$, similar measure factors including non-trivial one-loop determinants appeared in \cite{Schwarz:1979ae, Pestun:2005rp}.

Apart from the above, the proof of Freed and Quinn \cite{Freed:1991bn, Freed:1994ad} showing that ordinary Dijkgraaf-Witten theory defines a field theory functor generalizes easily. 

We define the higher abelian Dijkgraaf-Witten theories only as field theories extended to codimension 2, because we do not have a clear picture of the higher categorical objects assigned by the field theory functor to manifolds of codimension higher than 2. The heuristic arguments of \cite{Freed:1994ad} suggest however that there should be no problem defining these theories as fully extended field theories.

After publication, we realized that as non-extended field theories, the theories constructed here are a special case of a general construction by Quinn described in \cite{Quinn1991}. In this construction, the classifying space of fields $B\Gamma$ appearing in Dijkgraaf-Witten theory is replaced by any topolgical space with finite homotopy groups. We should also mention that closely related field theories have been constructed by \v{S}evera in \cite{Severa:2002qe}. The sketch of a general framework for finite path integration was presented in \cite{2009arXiv0905.0731F}.

It would be interesting to construct Dijkgraaf-Witten theories of higher gauge fields valued in non-abelian finite groups. The quantum theory of non-abelian higher gauge fields is unknown\footnote{See however \cite{Zucchini:2011aa} for an approach to quantization using the BV formalism.}, and the latter appear in several physically interesting theories, such as (2,0) superconformal field theories in six dimensions or gauged supergravities. One may hope that the simple setting of Dijkgraaf-Witten theory will provide interesting insights. A possible avenue is to repeat the present construction in the framework of non-abelian cohomology (see for instance \cite{giraud1971cohomologie, Toen, 2006math.....11317B}). In the context of state sum models, results have been obtained by Yetter in \cite{Yetter:1992rz} (see also \cite{2006math......8484F}) in the case $p = 2$, and by Porter in \cite{JLM:20235} for generic $p$. We will not discuss further the non-abelian case here.

The paper is organized as follows. In Section \ref{SecGamGerb}, we explain that the isomorphism classes of fields in the higher DW theories are classified by the $p$th cohomology group of the underlying manifold with value in $\Gamma$. In Section \ref{SecStruct}, we describe the structures on spacetime manifolds required to define the theory. We describe in Section \ref{SecFields} the space of fields over a manifold, paying particular attention to the case where the latter has a boundary. We define there the measure factors crucial for the definition of the theory in Section \ref{SecDef}. In Section \ref{SecGlu}, we show that the field theory functor is compatible with the gluing of manifolds.

\section{Degree $p$ $\Gamma$-valued gauge fields}

\label{SecGamGerb}

Let $\Gamma$ be a finite abelian group. We would like to construct a version of DW theory in which the fields on which the path integral is performed generalize principal $\Gamma$-bundles in the same way as $p$-form gauge fields generalize ordinary (i.e. 1-form) abelian gauge fields. While the case $p = 2$ is rather well understood \cite{Baez:2004in, Schreiber2011, Zucchini:2011aa, Jurco:2014mva}, we do not have a good picture for such objects for general $p$. However, we can make sense of their isomorphism classes as follows, which will turn out to be sufficient to formulate the DW theory. 

We remark that the isomorphism classes of principal $\Gamma$-bundles over a manifold $M$ are classified by $H^1(M;\Gamma)$, which is ultimately due to the fact that the classifying space $B\Gamma$ is an Eilenberg-MacLane space $K(\Gamma,1)$. The usual DW theory can be reformulated in terms of degree 1 $\Gamma$-valued cocycles instead of principal $\Gamma$-bundles. The precise model used for the cochains has no influence on what follows and we take singular cochains for definiteness. Of course, there is no bijection between principal $\Gamma$-bundles and $\Gamma$-valued 1-cocycles, but there is a bijection between the isomorphism classes of such objects. As the DW theory relies ultimately only on gauge invariant data, the two formulations are equivalent. \footnote{To be more precise, as pointed out by the referee, we need the groupoids of fields to be equivalent as categories. I.e. the spaces of gauge transformations (morphisms) should be in bijection as well. In the case of interest to us, this condition is trivially satisfied. It is sufficient to check that in both models, the automorphism group of a field is given by $H^0(M; \Gamma)$.} This is a concrete illustration of the fact, well-known to physicists, that a "gauge symmetry" is only a redundancy in the description of the theory, and not a property of the theory itself. 

In the cocycle formulation, the generalization to higher degree is obvious. The fields of the higher abelian DW theories are degree $p$ $\Gamma$-valued cocycles, and the equivalence classes of gauge fields are given by the degree $p$ cohomology groups valued in $\Gamma$. We now develop this picture more systematically.%We declare that any two such cocycles are isomorphic if they differ by the differential of a degree $p-1$ $\Gamma$-valued cochain. The isomorphism classes of fields on a closed manifold $M$ are therefore elements of the cohomology group $H^p(M;\Gamma)$.

Let $M$ be an closed oriented manifold. The fields on $M$ are degree $p$ $\Gamma$-valued cocycles, which we write hatted. A cocycle $\hat{P}_1$ is isomorphic to a cocycle $\hat{P}_2$ if they define the same cohomology class, i.e. if there is a degree $p-1$ cochain $\hat{\phi}$ such that $\hat{P}_2 = \hat{P}_1 + d\hat{\phi}$. As they have the same action on cocycles, we identify isomorphisms differing by the differential of a cochain, i.e. $\hat{\phi} \sim \hat{\phi} + d\hat{\rho}$. With these identifications, the automorphism group ${\rm Aut}(\hat{P})$ is $H^{p-1}(M; \Gamma)$, which is a finite group. We write $P$ for the cohomology class of $\hat{P}$.

We extend the discussion to an oriented manifold $M$ with boundary or corners, for which we need the notion of relative cocycle. Let $\hat{Q}$ be a degree $p$ $\Gamma$-valued cocycle over $\partial M$. A degree $p$ $\Gamma$-valued cocycle on $M$ relative to $\hat{Q}$ (in short a relative cocycle), is a pair $(\hat{P},\hat{\theta})$ where $\hat{P}$ is a degree $p$ $\Gamma$-valued cocycle on $M$ and $\hat{\theta}$ is a degree $p-1$ $\Gamma$-valued cochain on $\partial M$ such that $\hat{P}|_{\partial M} = \hat{Q} + d\hat{\theta}$. An isomorphism between two relative cocycles $(\hat{P}_1,\hat{\theta}_1)$ and $(\hat{P}_2,\hat{\theta}_2)$ is an equivalence class of degree $p-1$ $\Gamma$-valued cochain $\hat{\phi}$ on $M$ such that $\hat{P}_2 = \hat{P}_1 + d\hat{\phi}$ and $\hat{\theta}_2 = \hat{\theta}_1 + \hat{\phi}|_{\partial M}$. Two such cochains are equivalent if they differ by the differential of a cochain vanishing on the boundary: $\hat{\phi} \sim \hat{\phi} + d\hat{\rho}$ with $\hat{\rho}|_{\partial M} = 0$. The automorphism group ${\rm Aut}(\hat{P},\hat{\theta})$ is $H^{p-1}(M,\partial M; \Gamma)$, the relative cohomology group with value in $\Gamma$, which is a finite group. We write $(P,\theta)$ for the cohomology class of $(\hat{P},\hat{\theta})$.

We introduce the following notation. We write $F(M,\hat{Q})$ for the groupoid of degree p $\Gamma$-valued cocycles on $M$ relative to the cocycle $\hat{Q}$ on $\partial M$. When $M$ has no boundary, we simply write $F(M) = H^p(M;\Gamma)$. We will write $E(M,\hat{Q})$ for the group of equivalence classes of the groupoid $F(M,\hat{Q})$. $E(M,\hat{Q})$ is a torsor on $H^p(M,\partial M;\Gamma) = E(M,0)$.  

\section{Structures on manifolds} 

\label{SecStruct}

We consider manifolds endowed with certain unspecified geometrical/topological structures, denoted by $\mathsf{F}$ (see Appendix A.4 of \cite{Monnierd}). We assume that given a manifold $M$, possibly with boundary, $\mathsf{F}(M)$ includes an orientation on $M$ and a map $\gamma_P$ from $M$ to $K(\Gamma,p)$. Upon picking a universal choice of cocycle $\hat{U}$ representing a generator of $H^p(K(\Gamma,p);\Gamma)$, $\gamma_P$ determines an element $\hat{P}$ of $E(M,\hat{Q})$, with $\hat{Q} = (\gamma_P)^\ast(\hat{U})|_{\partial M}$. Hence $\gamma_P$ determines a gauge field on $M$. We will call such manifolds \emph{manifolds with $\mathsf{F}$-structure}, or simply \emph{$\mathsf{F}$-manifolds}. We write $\bar{\mathsf{F}}$ for the structure encoding the same data as $\mathsf{F}$, except for the map $\gamma_P$. We also assume that we are given a cocycle $\hat{c}_U$ representing a cohomology class $c_U \in H^d(K(\Gamma,p), U(1))$, that plays the role of the exponentiated action of the theory. The data $\mathsf{F}(M)$ then includes a cocycle $\hat{c} := \gamma_P^\ast \hat{c}_U \in H^d(M, U(1))$.

As explained in \cite{Freed:1991bn, Freed:1994ad}, there is a sense in which one can integrate $c$ over the $d-k$-dimensional manifold $M$. For $k = 0$, the integration map is the usual integration of cochains, yielding an element of $U(1)$. For $k = 1$, one obtains a Hermitian line, i.e. a 1-dimensional Hilbert space. For $k = 2$, one obtain a 2-Hermitian line, which is a category equivalent to the category $\mathcal{H}_1$ of finite dimensional Hilbert spaces (see for instance Appendix A.2 of \cite{Monnierd}). For higher $k$, one obtain higher analogues of Hermitian lines \cite{Freed:1994ad}. We will write $\mathcal{I}_c$ for the integration map.

$\mathcal{I}_c$ is a field theory defined on manifolds with $\mathsf{F}$-structure. Its value depends only on the homotopy class of the map $\gamma_P$. It can be seen as a classical version of the DW theory \cite{Freed:1991bn, Freed:1994ad}. More precisely, in the terminology of geometric quantization, it is the prequantum version of the DW theory determined by the exponentiated action $c_U$. The quantum DW theory $\mathcal{D\!\!W}_{c}$ is defined on manifolds with $\bar{\mathsf{F}}$-structure, via a sum of $\mathcal{I}_c$ over the space of isomorphism classes of degree $p$ $\Gamma$-valued gauge fields. This sum should be interpreted as a path integral over the field space of the theory.

In the following, all the manifolds are assumed to be $\bar{\mathsf{F}}$-manifolds, and we denote $\mathsf{F}$-manifolds by pairs $(M,\hat{P})$, where $M$ is a $\bar{\mathsf{F}}$-manifold and $\hat{P}$ is the gauge field encoded in $\mathsf{F}(M)$. Cocycles/gauge fields are always hatted and their cohomology classes/equivalence class are denoted by the same letter without a hat.

\section{Measure} 

\label{SecFields}

We now define measure factors that play a crucial role in the definition of the theory, and prove a fundamental identity they satisfy. Let
\be
\label{EqDefNormFact1}
\mu_M = \prod_{i = 0}^{p-1} \left| H^i(M,\partial M;\Gamma) \right|^{(-1)^{p-i}} \;,
\ee
where $|G|$ denotes the order of the finite group $G$. Let us furthermore define for $N \subset M$, $N \cap \partial M = \emptyset$,
\be
\label{EqDefNormFact2}
\mu_{(M,N)} = \prod_{i = 0}^{p-1} \left| H^i(M,N \cup \partial M;\Gamma) \right|^{(-1)^{p-i}} \;.
\ee
Let $K$ be the kernel of the map $H^p(M,N \cup \partial M;\Gamma) \rightarrow H^p(M,\partial M;\Gamma)$.
\begin{lemma}
The following equality holds:
\be
\label{EqRelMeasAutGroups}
\mu_M = |K| \mu_{(M,N)} \mu_{N}
\ee
\end{lemma}
\begin{proof}
This is an immediate consequence of the long exact sequence for relative cohomology:
\begin{align}
\label{EqLongExSeqRelCohom}
... \rightarrow H^{p-2}(N) & \rightarrow H^{p-1}(M,N \cup \partial M)  \rightarrow H^{p-1}(M,\partial M) \\
& \rightarrow H^{p-1}(N) \rightarrow H^p(M,N \cup \partial M) \rightarrow H^p(M,\partial M) \notag \;,
\end{align}
where we suppressed the argument $\Gamma$ in the cohomology groups.
\end{proof}

Remark that in the ordinary Dijkgraaf-Witten theory, $\mu_{(M,\mathscr{P})} = 1/|{\rm Aut}(\mathscr{P})|$, where $\mathscr{P}$ is a principal $\Gamma$-bundle, and ${\rm Aut}(\mathscr{P})$ is the group of automorphisms of $\mathscr{P}$ leaving $\mathscr{P}|_{\partial M}$  fixed. When $\Gamma$ is abelian, $|{\rm Aut}(\mathscr{P})| = |H^0(M,\partial M; \Gamma)|$, which is consistent with \eqref{EqDefNormFact1}.

%Let $\mathcal{C}(M,Q)$ be the category whose objects are degree $p$ $\Gamma$-valued cocycles on $M$ relative to $Q$ and whose morphisms are isomorphisms between relative cocycles. We will write simply $\mathcal{C}(M)$ when the boundary of $M$ is empty.

%$\mathcal{C}(M,Q)$ is in fact a groupoid, and the corresponding set of equivalence classes will be written $\bar{\mathcal{C}}(M,Q)$. $[P,\theta]$ stands for the equivalence class of $(P,\theta)$ (written $[P]$ if $M$ has no boundary).  Define $\mu_{[P,\theta]} := 1/|{\rm Aut}(P,\theta)|$, the inverse of the order of the automorphism group of $(P,\theta)$ in $\mathcal{C}(M,Q)$, which depends only on the isomorphism class $[P,\theta]$.

\section{Definition of the theory}

\label{SecDef}

In the following we use the following conventions. A 0-Hilbert space is a complex number. A 1-Hilbert space is a finite dimensional Hilbert space. The category of 1-Hilbert spaces is denoted by $\mathcal{H}_1$. A 2-Hilbert space \cite{1996q.alg.....9018B} is a $\mathbb{C}$-linear category linearly equivalent to the $n$th Cartesian product of $\mathcal{H}_1$ with itself, endowed with extra structure, see also Appendix A.2 of \cite{Monnierd}. In particular, a 2-Hilbert space $H$ is endowed with a functor $(\bullet, \bullet)_H: H^{\rm op} \times H \rightarrow \mathcal{H}_1$, playing the role of the inner product. The 2-Hilbert spaces form a 2-category $\mathcal{H}_2$. $\mathcal{H}_2$ admits a dagger structure given by the complex conjugation and a symmetric monoidal structure described in Section 4.4 of \cite{1996q.alg.....9018B}.

We write $M^{d,p}$ for a generic $\bar{\mathsf{F}}$-manifold of dimension $d$ with corners of dimension $d-p$ or higher. If $\mathsf{X} = \mathsf{F}, \bar{\mathsf{F}}$, $\mathcal{B}^{d,p}_\mathsf{X}$ is the bordism category consisting of $\mathsf{X}$-manifolds of dimension $d-p$, ..., $d$ with corners of dimension $d-p$ or higher, see Appendix A.4 of \cite{Monnierd}. The bordism category has a dagger structure given by the orientation reversal of manifolds, and a symmetric monoidal structure given by the disjoint union of manifolds. 

We will define below the quantum DW theory as a 2-functor
\be
\mathcal{D\!\!W}_{c}: \mathcal{B}^{d,2}_{\bar{\mathsf{F}}} \rightarrow \mathcal{H}_2
\ee
compatible with the dagger and the monoidal structures. 

\paragraph{Prequantum theory} We rely on the fact that the prequantum DW theory 
\be
\mathcal{I}_{c}: \mathcal{B}^{d,2}_{\mathsf{F}} \rightarrow \mathcal{H}_2
\ee
is such a 2-functor \cite{Freed:1991bn, Freed:1994ad}. (See also Section 4 of \cite{Monnierd}.) For $k = 0,1,2$, the prequantum DW theory $\mathcal{I}_{c}$ associates a (one-dimensional) $k$-Hilbert space $\mathcal{I}_{c}(M^{d-k},P)$ to a closed $d-k$-dimensional $\mathsf{F}$-manifold $(M^{d-k},\hat{P})$. For $k = 0,1$, $\mathcal{I}_{c}$ associates a vector $\mathcal{I}_{c}(M^{d-k,1},P)$ in $\mathcal{I}_{c}(\partial M^{d-k,1},Q)$ to each manifold $d-k$-dimensional $\mathsf{F}$-manifold $M^{d-k,1}$ with boundary endowed with $\hat{P} \in F(M,\hat{Q})$. 

%For $k = 0,1,2$, the prequantum DW theory $\mathcal{I}_{c}$ associates a $k$-Hilbert space $\mathcal{I}_{c}(M^{d-k},P)$ to a closed $d-k$-dimensional $\mathsf{F}$-manifold $(M^{d-k},P)$. More precisely, $\mathcal{I}_{c}$ associates a $k$-Hilbert space $\mathcal{I}_{c}(M^d,\hat{P})$ to each closed $d-k$-dimensional $\mathsf{F}$-manifold $(M^{d-k},P)$ endowed with a cocycle $\hat{P}$ representing $P$, together with canonical isomorphisms between $\mathcal{I}_{c}(M^{d-k},\hat{P})$ and $\mathcal{I}_{c}(M^{d-k},\hat{P}')$ whenever $P = P'$. An inverse limit constrution yields the canonical $k$-Hilbert space $\mathcal{I}_{c}(M^d,P)$ from this data.
%
%For $k = 0,1$, $\mathcal{I}_{c}$ associates a vector $\mathcal{I}_{c}(M^{d-k,1},P)$ in $\mathcal{I}_{c}(\partial M^{d-k,1},P|_{\partial M^{d-k,1}})$ to each manifold $d-k$-dimensional $\mathsf{F}$-manifold $M^{d-k,1}$ with boundary. More precisely, $\mathcal{I}_{c}$ provides vectors $\mathcal{I}_{c}(M^{d-k,1},\hat{P}) \in \mathcal{I}_{c}(\partial M^{d-k,1},\hat{P}|_{\partial M^{d-k,1}})$ which are equivariant with respect to the canonical isomorphisms above, inducing a vector in $\mathcal{I}_{c}(\partial M^{d-k,1},P|_{\partial M^{d-k,1}})$.

\paragraph{Closed $d-k$-dimensional manifolds} Here $k = 0,1,2$. We define the value of the quantum DW theory on $M^{d-k}$ by 
\be
\label{EqQDWClMan}
\mathcal{D\!\!W}_{c}(M^{d-k}) = \sum_{P \in E(M)} \mu_{M^{d-k}} \mathcal{I}_{c}(M^{d-k}, P) \;.
\ee
The sum sign should be understood as an ordinary sum when $k = 0$, as a direct sum of Hilbert spaces when $k = 1$ and as the direct sum of 2-Hilbert spaces for $k = 2$ (see Appendix A.2 in \cite{Monnierd}). The multiplication by $\mu_{M^{d-k}}$ also deserves an explanation. For $k = 0$ this is the ordinary multiplication of complex numbers by the rational number $\mu_{M^{d-k}}$. For $k = 1$, $\mu \in \mathbb{Q}_+$ and $H$ a Hilbert space, $\mu H$ is the vector space $H$, endowed with the inner product of $H$ rescaled by $\mu$: $(\bullet, \bullet)_{\mu H} = \mu (\bullet, \bullet)_{H}$. For $k = 2$, let $H$ be a 2-Hilbert space, endowed with an inner product $(\bullet, \bullet)_H$ valued in $\mathcal{H}_1$. Then $\mu H$ is the 2-vector space $H$, endowed with an inner product $(\bullet, \bullet)_{\mu H}$ defined as follows. For any $V_1, V_2 \in H$, $(V_1,V_2)_{\mu H} = \mu (V_1, V_2)_{H}$, where the multiplication on the right-hand side should be interpreted according to the $k=1$ case we described above.

\paragraph{$d-k$-dimensional manifolds with boundary} Here $k = 0,1$. We define the value of the quantum DW theory on $M^{d-k,1}$ by:
\be
\label{EqQDWManBound}
\mathcal{D\!\!W}_{c}(M^{d-k,1}) = \sum_{Q \in E(\partial M^{d-k,1})} \sum_{P \in E(M^{d-k,1}, \hat{Q})} \mu_{M^{d-k,1}} \mathcal{I}_{c}(M^{d-k,1},P) \;.
\ee
Here $\hat{Q}$ is any cocycle representing the class $Q$. The terms on the right-hand side of \eqref{EqQDWManBound} are understood as elements of $\mathcal{I}_{c}(\partial M^{d-k,1},Q)$.

Consistency requires that 
\be
\label{EqConstManBound}
\mathcal{D\!\!W}_{c}(M^{d-k,1}) \in \mathcal{D\!\!W}_{c}(\partial M^{d-k,1}) \;.
\ee
But this is immediately implied by the corresponding relation for the prequantum DW theory: $\mathcal{I}_{c}(M^{d-k,1},P) \in \mathcal{I}_{c}(\partial M^{d-k,1},P|_{\partial M^{d-k,1}})$ \cite{Freed:1991bn, Freed:1994ad}. \eqref{EqConstManBound} implies in particular that given a bordism $B^{d-k}$ between manifolds $M_1^{d-k-1}$ and $M_2^{d-k-1}$, $\mathcal{D\!\!W}_{c}(B^{d-k})$ is a homomorphism ($k = 0$) or a $\mathbb{C}$-linear functor ($k = 1$) from $\mathcal{D\!\!W}_{c}(M_1^{d-k-1})$ to $\mathcal{D\!\!W}_{c}(M_2^{d-k-1})$.

\paragraph{$d$-dimensional manifolds with corners} Let $M^{d,2}$ be a $d$-dimensional manifold with $\partial M^{d,2} = -N_1 \cup N_2$, where $\partial N_1 = \partial N_2 = -M_1 \sqcup M_2$. We define analogously to \eqref{EqQDWClMan} and \eqref{EqQDWManBound}
\be
\label{EqQDWManCor}
\mathcal{D\!\!W}_{c}(M^{d,2}) = \sum_{R \in E(-M_1 \sqcup M_2)} \sum_{Q \in E(-N_1 \cup N_2,\hat{R})} \sum_{P \in E(M^{d,2}, \hat{Q})} \mu_{M^{d,2}} \mathcal{I}_{c}(M^{d,2},P) \;.
\ee
As before, we picked arbitrary cocycle lifts $\hat{R}$ and $\hat{Q}$ of $R$ and $Q$. $F(-N_1 \cup N_2,\hat{R})$ is the groupoid of cocycles restricting to $\hat{R}$ on $-M_1 \sqcup M_2 \subset -N_1 \cup N_2$, with arrows given by the addition of exact cocycles vanishing on $-M_1 \sqcup M_2$. $E(-N_1 \cup N_2,\hat{R})$ is the corresponding set of equivalence classes.

The fact that $\mathcal{D\!\!W}_{c}(M^{d,2})$ is a 2-morphism between the 1-morphisms $\mathcal{D\!\!W}_{c}(N_1)$ and $\mathcal{D\!\!W}_{c}(N_2)$ is directly inherited from the corresponding property of the prequantum DW theory \cite{Monnierd}.

\paragraph{Higher codimension} Formulas \eqref{EqQDWClMan}, \eqref{EqQDWManBound} and \eqref{EqQDWManCor} clearly have the same structure. Given a concrete construction of the prequantum DW field theory as a fully extended field theory, for instance along the lines proposed in \cite{Fiorenza:2012ec}, the obvious generalization of these formulas should define the higher abelian DW theories as fully extended field theories. We expect the proof of the gluing law in the next section to be formally identical, see \cite{Freed:1994ad} for the case of ordinary DW theories.

\paragraph{Compatibility} The compatibility of $\mathcal{D\!\!W}_{c}$ with the dagger and monoidal structures of $\mathcal{B}^{d,2}_{\bar{\mathsf{F}}}$ and $\mathcal{H}_2$ comes from the compatibility of $\mathcal{I}_{c}$ with these structures \cite{Freed:1991bn, Freed:1994ad, Monnierd}, and the fact that $\mu(M_1 \cup M_2) = \mu(M_1) \mu(M_2)$ for $M_1$ and $M_2$ disjoint manifolds.

\section{Gluing} 

\label{SecGlu}

The compatibility of the prequantum DW theory with gluing (i.e. the compatibility of the functor $\mathcal{I}_{c}$ with the composition of morphisms in $\mathcal{B}^{d,2}_\mathsf{F}$ and $\mathcal{H}_2$) is obvious from the locality of the integral. Because of the sums involved, the compatibility with gluing is not obvious in the DW theory and we check it here. 

Let $M^{d-k,1}$ be as usual a $d-k$-dimensional $\bar{\mathsf{F}}$-manifold with boundary and let $N \subset M^{d-k,1}$ be a codimension 1 submanifold disjoint from the boundary. Let $M_N^{d-k,1}$ be the manifold $M^{d-k,1}$ cut along $N$, whose boundary is $\partial M^{d-k,1} \cup N \cup -N$. The compatibility with gluing is equivalent to the following
\begin{theorem}
\label{ThCompGlueDW}
We have:
\be
\label{EqCompGlueDW}
\mathcal{D\!\!W}_{c}(M^{d-k,1}) = {\rm Tr}_{\mathcal{D\!\!W}_{c}(N)}\left(\mathcal{D\!\!W}_{c}(M_N^{d-k,1})\right) \;,
\ee
where ${\rm Tr}$ on the right hand side denotes the contraction of
\be
\mathcal{D\!\!W}_{c}(M_N^{d-k,1}) \in \mathcal{D\!\!W}_{c}(\partial M^{d-k,1}) \otimes \mathcal{D\!\!W}_{c}(N) \otimes \left(\mathcal{D\!\!W}_{c}(N)\right)^\dagger
\ee
using the canonical pairing between $\mathcal{D\!\!W}_{c}(N)$ and its dual. 
\end{theorem}
Remark that the trace involves a scalar multiplication. For $k = 1$, the pairing is valued in $\mathcal{H}_1$ and the scalar multiplication is a tensor product-like operation between a Hilbert space and an element of the 2-Hilbert space $\mathcal{D\!\!W}_{c}(\partial M^{d-k,1})$, see for instance Appendix A.2 of \cite{Monnierd}. 

Our proof of the gluing relation \eqref{EqCompGlueDW} is strongly inspired by the corresponding proof in \cite{Freed:1991bn, Freed:1994ad}, valid for the usual DW theory and its extended version. In the present proof, we write $M$ for $M^{d-k,1}$ and $M_N$ for $M_N^{d-k,1}$. 
%We omit the argument $\Gamma$ in all cohomology groups to simplify the notation. All cohomology groups are understood to be relative with respect to $\partial M \subset M, M_N$, and we suppress as well this information from the notation. 
Let $\pi: M_N \rightarrow M$ be the gluing map that identifies the components $N$ and $-N$ of the boundary of $M_N$. 

\begin{proof}[Proof of Theorem \ref{ThCompGlueDW}]  We use the definition of the left-hand side to write 
\be
\label{EqProofGlueDW1}
\mathcal{D\!\!W}_{c}(M) = \mu_M \sum_{R \in E(\partial M)} \sum_{P \in E(M,\hat{R})} \mathcal{I}_{c}(M, P) \;.
\ee
We can perform the analysis term by term with respect to the first sum, so we fix $R \in E(\partial M)$. Let $E_N(M,\hat{R})$ and $E_{\rm ext}(N,\hat{R})$ be the kernel and image of the restriction map $E(M,\hat{R}) \rightarrow E(N)$, and let $Q_{\rm ext}$ be a choice of preimage for each $Q \in E_{\rm ext}(N,\hat{R})$. We decompose the sum over the classes in $E(M,\hat{R})$ as a sum over $E_N(M,\hat{R})$ and $E_{\rm ext}(N,\hat{R})$. The right-hand side of \eqref{EqProofGlueDW1}, for fixed $R$, becomes
\be
\label{EqProofGlueDW2}
\mu_{M} \sum_{Q \in E_{\rm ext}(N,\hat{R})} \sum_{P \in E_N(M,\hat{R})} \mathcal{I}_{c}(M, P + Q_{\rm ext}) \;.
\ee
We use \eqref{EqRelMeasAutGroups}, the gluing relation for $\mathcal{I}_{c}$ and the linearity of the trace to rewrite \eqref{EqProofGlueDW2} as
\be
\label{EqProofGlueDW3}
\mu_N  \mu_{(M,N)} |K| \sum_{Q \in E_{\rm ext}(N,\hat{R})} {\rm Tr}_{\mathcal{I}_{c}(N,Q)} \left( \sum_{P \in E_N(M,\hat{R})} \mathcal{I}_{c}(M_N,\pi^\ast (P + Q_{\rm ext})) \right) \;.
\ee
Let $E_{\rm ext}(N \sqcup -N,\hat{R})$ be the subset of $E(N \sqcup -N)$ consisting of equivalence classes of fields admitting an extension to $M_N$ restricting to $\hat{R}$ on $\partial M$. Let us choose an element $Q'_{\rm ext} \in E(M_N, \hat{R} + \hat{Q}')$ for each $Q' \in E(N \cup -N)$. Because the trace selects the diagonal component, we can replace the sum over $Q \in E_{\rm ext}(N)$ in \eqref{EqProofGlueDW3} by a sum over $Q' \in E_{ext}(N \cup -N)$. %\comment{As we are not immediately introducing the trace over $\mathcal{D\!\!W}_{c}(N)$, we have to restrict ourselves to the "diagonal subspace", such that the bundle on the two components of $N \cup -N$ glue consistently.} 

Next, excising a tubular neighborhood of $N$ in $M$, we remark that 
\be
\label{EqIsomCohomMNRelCohom}
E(M_N, 0) = H^p(M_M,\partial M_N;\Gamma) \simeq H^p(M,N \cup \partial M;\Gamma) \;,
\ee
where the latter isomorphism is given by excision. We have therefore a surjective homomorphism $E(M_N,0) \rightarrow E_N(M,0)$, and hence also (non-canonical) surjective homomorphisms $E(M_N, \hat{R} + \hat{Q}') \rightarrow E_N(M,\hat{R})$. Given \eqref{EqIsomCohomMNRelCohom}, the order of the kernel of these homomorphisms is $|K|$. This means that we can replace the sum $|K|\sum_{P \in E_N(M,\hat{R})}$ by $\sum_{P' \in E_N(M_N,\hat{R}+ \hat{Q}')}$. We obtain therefore
\be
\label{EqProofGlueDW4}
\mu_N  \mu_{(M,N)} \sum_{Q' \in E_{ext}(N \cup -N,\hat{R})} {\rm Tr}_{\mathcal{I}_{c}(N,Q')} \left( \sum_{P' \in E(M_N,\hat{R} + \hat{Q}')} \mathcal{I}_{c}(M_N,P') \right) \;.
\ee
But now we use again the linearity of the trace, the fact that $\mu_{(M,N)} = \mu_{M_N}$ and we remark that after summing over $R$, we obtain $\mathcal{D\!\!W}_{c}(M_N)$. Moreover, the trace over $\mathcal{D\!\!W}_{c}(N)$ is $\mu_N$ times the sum of the traces over $\mathcal{I}_{c}(N,Q)$, so we finally obtain \eqref{EqCompGlueDW}.
\end{proof}

This proves that $\mathcal{D\!\!W}_{c}$ is a 2-functor, hence defines a field theory.

\subsection*{Acknowledgments}

This research is supported in part by Forschungskredit FK-14-108 of the University of Zürich, SNF Grant No.200020-149150/1 and by NCCR SwissMAP, funded by the Swiss National Science Foundation.

{
\small
%\bibliographystyle{../../Bibliographie/BibStyle/utphys}
%\bibliography{../../Bibliographie/References}
%\bibliographystyle{utphys}
%\bibliography{References}

\providecommand{\href}[2]{#2}\begingroup\raggedright\endgroup

}

\end{document}